\documentclass[letterpaper,11pt]{article}
\usepackage{amsfonts,amsmath,amssymb,amsthm}
\usepackage[noadjust]{cite}
\usepackage{comment,enumerate,hyperref}

 \newtheorem{theorem}{Theorem}[section]
 \newtheorem{corollary}[theorem]{Corollary}
 
 \newtheorem{lemma}[theorem]{Lemma}
 
 \newtheorem{conjecture}[theorem]{Conjecture}
\usepackage[margin=1in]{geometry}

 \theoremstyle{definition}
 \newtheorem{definition}{Definition}

\usepackage{tikz}

\usepackage{thm-restate}

\usepackage{xcolor,xspace}
\usepackage{booktabs}

\usepackage{algorithm}
\usepackage[noend]{algorithmic}

\usepackage{csquotes}

\newcommand{\poly}{\text{poly}}

\newcommand{\dist}{d}

\title{Partially Optimal Edge Fault-Tolerant Spanners}

\author{Greg Bodwin\\University of Michigan\\ \texttt{bodwin@umich.edu} \and Michael Dinitz\thanks{Supported in part by NSF award CCF-1909111.}\\Johns Hopkins University\\ \texttt{mdinitz@cs.jhu.edu} \and Caleb Robelle\\UMBC\\ \texttt{carobel1@umbc.edu}}
\begin{document}
\maketitle

\begin{abstract}
    Recent work has established that, for every positive integer $k$, every $n$-node graph has a $(2k-1)$-spanner on $O(f^{1-1/k} n^{1+1/k})$ edges that is resilient to $f$ edge or vertex faults.  For \emph{vertex} faults, this bound is tight.  However, the case of \emph{edge} faults is not as well understood: the best known lower bound for general $k$ is $\Omega(f^{\frac12 - \frac{1}{2k}} n^{1+1/k} +fn)$.  Our main result is to nearly close this gap with an improved upper bound, thus separating the cases of edge and vertex faults.  For odd $k$, our new upper bound is $O_k(f^{\frac12 - \frac{1}{2k}} n^{1+1/k} + fn)$, which is tight up to hidden $\poly(k)$ factors.  For even $k$, our new upper bound is $O_k(f^{1/2} n^{1+1/k} +fn)$, which leaves a gap of $\poly(k) f^{1/(2k)}$.  Our proof is an analysis of the fault-tolerant greedy algorithm, which requires exponential time, but we also show that there is a polynomial-time algorithm which creates edge fault tolerant spanners that are larger only by factors of $k$.
\end{abstract}
\thispagestyle{empty}

\clearpage
\setcounter{page}{1}

\section{Introduction}

% if we include this then maybe at least the TCS whispers guy will cite us

Let $G = (V, E)$ be a graph, possibly with edge lengths $w : E \rightarrow \mathbb{R}_{\geq 0}$. A $t$-spanner of $G$, for $t \geq 1$, is a subgraph $H = (V, E')$ that preserves all pairwise distances within a factor of $t$, i.e.,
% \begin{align*}
% d_{H}(u,v) \leq t \cdot d_G(u,v) \label{eq:nonfttest}
% \end{align*}
$$d_{H}(u,v) \leq t \cdot d_G(u,v)$$
for all $u,v \in V$ (where $d_{X}$ denotes the shortest-path distance in a graph $X$).  Since $H$ is a subgraph of $G$ it is also true that $d_G(u,v) \leq d_{H}(u,v)$, and so distances in $H$ are the same as in $G$ up to a factor of $t$.  The distance preservation factor $t$ is called the \emph{stretch} of the spanner.  
Spanners were introduced by Peleg and Ullman~\cite{PelegU:89} and Peleg and Sch{\"{a}}ffer~\cite{PelegS:89}, and have a wide range of applications in routing \cite{PelegU:89-routing}, synchronizers \cite{awerbuch1990network}, broadcasting \cite{awerbuch1991cient,peleg2000distributed}, distance oracles \cite{thorup2005approximate}, graph sparsifiers \cite{kapralov2012spectral}, preconditioning of linear systems \cite{elkin2008lower}, etc. 

The most common objective in spanners research is to achieve the best possible existential size-stretch trade-off.
Most notably, a landmark result of Alth\"ofer et al.~\cite{AlthoferDDJS:93} analyzed the following simple and natural greedy algorithm: given an $n$-node graph $G$ and an integer $k \geq 1$, consider the edges of $G$ in non-decreasing order of their weight and add an edge $(u,v)$ to the current spanner $H$ if and only if $\dist_{H}(u,v)>(2k-1) w(u,v)$.
They proved that this algorithm produces $(2k-1)$-spanners of \emph{existentially optimal size}.
In particular, the spanner produced has size $O(n^{1+1/k})$, and assuming the well-known \emph{Erd\H{o}s girth conjecture}~\cite{erdHos1964extremal} there are graphs in which \emph{every} $(2k-1)$ spanner (and in fact every $2k$-spanner) has at least $\Omega(n^{1+1/k})$ edges.  

A crucial aspect of real-life systems that is not captured by the standard notion of spanners is the possibility of \emph{failure}.
If some edges (e.g., communication links) or vertices (e.g., computer processors) fail, what remains of the spanner might not still approximate the distances in what remains of the original graph.
This motivates the notion of \emph{fault tolerant} spanners: 
\begin{definition} [Fault Tolerant Spanners] \label{def:FT}
A subgraph $H$ is an $f$-edge fault tolerant ($f$-EFT) $t$-spanner of $G = (V, E)$ if
% \begin{align*}
% d_{H \setminus F}(u,v) \leq t \cdot d_{G \setminus F}(u,v) %\label{eq:fttest}
% \end{align*}
$$d_{H \setminus F}(u,v) \leq t \cdot d_{G \setminus F}(u,v)$$
for all $u,v \in V$ and $F \subseteq E$ with $|F| \leq f$.
\end{definition}

In other words, an $f$-EFT spanner contains a spanner of $G\setminus F$ for every set of $|F| \le f$ edges that could fail.  The definition for vertex fault tolerance (VFT) is equivalent, with the only change being that $F \subseteq V \setminus \{u,v\}$.

Fault tolerant spanners were originally introduced in the setting of geometric graphs
%setting (where the vertices are points in $\mathbb{R}^d$ and the initial graph $G$ is the complete graph with Euclidean distances)
by Levcopoulos, Narasimhan, and Smid \cite{levcopoulos1998efficient} and have since been studied extensively in that setting~\cite{LNS98,lukovszki1999new,czumaj2004fault,NS07}. 
Chechik, Langberg, Peleg and Roditty \cite{ChechikLPR:10} were the first to study fault-tolerant spanners in general graphs, giving a construction of an $f$-VFT $(2k-1)$-spanner of size $O(f^2 k^{f+1} \cdot n^{1+1/k}\log^{1-1/k}n)$ and an
$f$-EFT $(2k-1)$-spanner of size $O(f\cdot n^{1+1/k})$.  So they showed that introducing tolerance to $f$ edge faults costs us an extra factor of $f$ in the size of the spanner, while introducing tolerance to $f$ vertex faults costs us a factor of roughly $f^2 k^{f+1}$ in the size (compared to the size of a non-fault tolerant spanner of the same stretch).
Dinitz and Krauthgamer \cite{DinitzK:11} later improved the extra factor paid in the VFT setting to about $f^3$.
At this point, it appeared that the EFT setting might be substantially easier than the VFT setting, in the sense that it allowed for a smaller dependence on $f$ in spanner size.
However, a recent series of papers has developed a set of techniques that apply equally well to both settings, yielding the same improved bounds for each~\cite{BDPW18,BP19,DR20,BDR21}.
This has culminated in the following theorem:
\begin{theorem} [FT upper bounds \cite{BDR21}] \label{thm:ftub}
There is a polynomial time algorithm that, given any positive integers $f, k$ and any $n$-node graph $G$, constructs an $f$-EFT or VFT $(2k-1)$-spanner of $G$ with $O(f^{1-1/k} n^{1+1/k})$ edges.
\end{theorem}

This upper bound on spanner size is notable because it fully matches the known lower bound for VFT spanners:
\begin{theorem} [VFT lower bounds \cite{BDPW18}] \label{thm:vftlb}
Assuming the girth conjecture~\cite{erdHos1964extremal}, for any positive integers $n, f, k$, there exist $n$-node graphs in which every $f$-VFT $(2k-1)$-spanner has $\Omega(f^{1-1/k} n^{1+1/k})$ edges.
\end{theorem}

Thus the VFT setting is now essentially completely understood.
However, there are still significant gaps in the EFT setting.
While the VFT lower bound from \cite{BDPW18} also holds for the EFT setting when $k=2$, for general $k$ it drops off considerably.
(We note that the $fn$ term in the following theorem statement is not explicitly given in~\cite{BDPW18}, but is a simple folklore result stating that any $f$-regular graph is the unique $f$-EFT spanner of itself):
%\mdnote{changed slightly}
\begin{theorem} [EFT lower bounds \cite{BDPW18}] \label{thm:eftlb}
Assuming the girth conjecture \cite{erdHos1964extremal}, for any positive integers $n, f, k$, there exist $n$-node graphs in which every $f$-EFT $(2k-1)$-spanner $H$ with
$$
|E(H)| = \begin{cases}\Omega\left(f^{1/2} n^{3/2}\right) & k = 2\\
\Omega\left(f^{1/2-1/(2k)} n^{1+1/k} + fn\right) & k > 2\end{cases}
$$
\end{theorem}

Hence there is a substantial $\poly(f)$ gap in the EFT setting when $k > 2$.  
%Worse, a \emph{technical} lower bound from \cite{BP19} demonstrates that the basic technique underlying recent progress, called \emph{blocking sets}, cannot be directly used to improve EFT bounds beyond what is currently known.
% an upper bound of $O(f^{1-1/k} n^{1+1/k})$ for both settings, including a polynomial time construction~\cite{BDR21}.  However, the only lower bound is from~\cite{BDPW18}.  In the VFT case, they proved that there are graphs in which every $f$-VFT $(2k-1)$-spanner has size at least $\Omega(f^{1-1/k} n^{1+1/k})$, and hence vertex fault tolerance is now completely understood (asymptotically).  For the EFT case, they proved two different lower bounds: for the special case of $k=2$ the lower bound is the same as for VFT ($\Omega(f^{1/2} n^{3/2})$), while for general $k$ the lower bound falls off significantly to $\Omega(f^{\frac12 - \frac{1}{2k}} n^{1+1/k})$.  So there is still a significant gap in the edge fault tolerant case when $k \geq 3$, and
So the main open question left is how to close this gap: how much extra do we have to pay to achieve edge fault tolerance compared to a non-fault tolerant spanner?  Is the right answer $f^{\frac12 - \frac{1}{2k}}$, or $f^{1 - 1/k}$, or somewhere in between?

\begin{table*}[t]
\begin{center}

    \begin{tabular}{llcl}
    \toprule
        \textbf{Spanner size} & \textbf{Tight?} & \textbf{Polytime?} & \textbf{Citation} \\
    \midrule
        $O \left(f \cdot n^{1+1/k} \right)$ & & \checkmark{}  & \cite{ChechikLPR:10} \\
        $O \left( \exp(k) f^{1 - 1/k} \cdot n^{1+1/k} \right) $  & $k=2$ only & & \cite{BDPW18} \\
        $O \left( f^{1 - 1/k} \cdot n^{1+1/k} \right)$ & $k=2$ only & & \cite{BP19} \\
        $O \left( k f^{1 - 1/k} \cdot n^{1+1/k} \right)$ & $k=2$ only & \checkmark{} & \cite{DR20} \\
        $O \left( f^{1 - 1/k} \cdot n^{1+1/k} \right)$ & $k=2$ only & \checkmark{} & \cite{BDR21}\\
        $O \left( k^{2} f^{1/2} \cdot n^{1+1/k} + k f n \right)$ for even $k$ & $k=2$ only ($^*$) & & \textbf{this paper} \\
        $O\left( k^{2} f^{1/2 - 1/(2k)} \cdot n^{1+1/k} + k f n \right)$ for odd $k$ & all fixed odd $k$ & & \textbf{this paper}\\
        $O \left( k^{5/2} f^{1/2} \cdot n^{1+1/k} + k^2 f n \right)$ for even $k$ & $k=2$ only ($^*$) & \checkmark{} & \textbf{this paper} \\
        $O\left( k^{5/2} f^{1/2 - 1/(2k)} \cdot n^{1+1/k} + k^2 f n \right)$ for odd $k$ & all fixed odd $k$ & \checkmark{} & \textbf{this paper}\\
    \bottomrule
    \end{tabular}
    
    \caption{\label{tab:priorwork} Prior work on the size of $f$-EFT $2k-1$ spanners, and our new bounds in this paper.  Tightness is implied by a lower bound from \cite{BDPW18} of $\Omega(f^{1/2-1/(2k)} n^{1+1/k})$, except for $k=2$ where the lower bound improves to $\Omega(f^{1/2} n^{1+1/k})$.  Except for $k \in \{2, 3, 5\}$, and for $k=7$ and large enough $f$ (see Theorem \ref{thm:uncondeftlb} to follow), this lower bound and our claims of tightness are conditional on the girth conjecture \cite{erdHos1964extremal}.  The ($^*$) indicates that we conjecture our bound for even $k$ to be tight for all fixed even $k$.}
    \end{center}
\end{table*}

\subsection{Our Results}

In this paper we nearly resolve this question by improving the upper bound:
%More formally, our main result is the following theorem:%\mdnote{space permitting, it might be nice to show a graph of our upper bound and the lower bound, with $k$ on the $x$ axis and the size (divided by $n^{1+1/k}$) on the $y$ axis (ignoring the polynomial factors of $k$}
\begin{theorem} [Main Result] \label{thm:main}
For any positive integers $n, f, k$, every $n$-node graph has an $f$-EFT $(2k-1)$-spanner $H$ with
$$
|E(H)| = \begin{cases} O\left(k^2 f^{1/2 - 1/(2k)} n^{1+1/k} +kfn\right) & k \text{ is odd}\\
O\left(k^2 f^{1/2} n^{1+1/k} + kfn \right) & k \text{ is even.}
\end{cases}
$$
\end{theorem}

%\gbnoteinline{Okay, here's a shot at a figure -- see if we like it.  Would rather shrink so it doesn't take up an entire page at some point.}

%To interpret this result, note that while the additional $kfn$ terms in Theorem~\ref{thm:main} may look like an additional loss, they are unavoidable up to the factor of $k$.
%While~\cite{BDPW18} may have only claimed lower bounds of $\Omega(f^{1/2} n^{3/2})$ for $k=2$ and $\Omega(f^{\frac12 - \frac{1}{2k}} n^{1+1/k})$ for general $k$,
%A folklore lower bound is that, for any $f$-regular graph, the only $f$-EFT spanner with finite stretch is the graph itself, and thus a dependence of $\Omega(fn)$ is needed.
%Previous upper bounds avoided this term only because their worse dependence on $f$ meant that it was dominated in the entire range of parameters by (the worse version of) the first term in our bound.  

Hence Theorem~\ref{thm:main} \emph{entirely resolves} the question of edge fault tolerance for constant odd $k$, and is off from the lower bound (Theorem~\ref{thm:eftlb}) by only quadratic factors of $k$ for nonconstant odd $k$, and only $k^2 f^{1/(2k)}$ for even $k$.  See Table~\ref{tab:priorwork} for the full context of our results, and see Figure~\ref{fig:results} in Appendix \ref{app:figures} for a visualization.

While it is interesting and important to optimize the dependence on $k$ (and close the gap between our new upper bounds and the known lower bounds), we remark that spanners are never used with $k$ larger than $O(\log n)$, since the additional stretch no longer yields additional sparsity.  Hence the factor of $k^2$ in our bounds is at most polylogarithmic in $n$.  On the other hand, the fault parameter $f$ can be significantly larger (for example, polynomial in $n$).  Hence the more central question in the area is to optimize dependence on $f$, and we are more concerned with our gap of $f^{1/(2k)}$ for even $k$ than our gap of $k^2$.  However, we conjecture that our upper bound is actually correct, and that our current construction and analysis are tight with respect to $f$:  %while we leave a gap of $f^{1/(2k)}$ for even $k$, we conjecture that our $f$-dependence is tight for even $k$ as well:
%\mdnote{combined and reordered some paragraphs}
\begin{conjecture}
For any positive integers $n, f$ and even integer $k$, there exist $n$-node graphs in which every $f$-EFT $(2k-1)$-spanner has $\Omega\left(f^{1/2} n^{1+1/k} \right)$ edges.
\end{conjecture}

Settling this conjecture, or generally making progress on the gap between the $f^{1/2}$ dependence in this work and the $f^{\frac12 - \frac{1}{2k}}$ lower bound on $f$-dependence from~\cite{BDPW18}, is the main open question left by this work.
We remark that there is some precedent in the literature for fundamentally different behavior for $(2k-1)$-spanners under even/odd $k$.
This occurs, for example, in the \emph{unbalanced Moore bounds} that control the maximum possible density of a bipartite graph with a very different number of nodes on its two sides \cite{hoory2002size}.
We believe that a similar even/odd effect may occur here, where the imbalance is essentially caused by the structure of edge faults.
The fact that the dependence $f^{1/2}$ is known to be tight for $k=2$ partially validates this intuition, and it may be ``simpler'' for the correct bounds to distinguish between even/odd $k$ rather than taking $k=2$ as a special case.

%However, we think it is likely possible to improve the dependence in our upper bounds on $k$.

\paragraph{Algorithmic Efficiency.} The spanner construction algorithm that we use to prove Theorem~\ref{thm:main} is the same greedy algorithm as in~\cite{BDPW18,BP19} (adapted for edge fault tolerance), which requires exponential time.  However, by combining our new analysis with the ideas used by~\cite{DR20}, we can obtain polynomial time at the price of a slightly worse dependence on $k$:
\begin{theorem} \label{thm:main-polytime}
There is a polynomial time algorithm that, given positive integers $f, k$ and an $n$-node input graph, outputs an $f$-EFT $(2k-1)$-spanner $H$ with
$$
|E(H)| = \begin{cases} O\left(k^{5/2} f^{1/2 - 1/(2k)} n^{1+1/k} + k^2 fn\right) & k \text{ is odd}\\
O\left(k^{5/2} f^{1/2} n^{1+1/k} + k^2 fn\right) & k \text{ is even.}
\end{cases}
$$
\end{theorem}

\section{Technical Overview} \label{sec:overview}%\mdnote{I think we should make this its own section}

Before diving into the details, we first give some intuition and context for our techniques.  

\paragraph{Background on Non-Faulty Spanners.}
In the non-faulty setting, the analysis of the greedy construction by Alth{\" o}fer et al.~\cite{AlthoferDDJS:93} passes through the \emph{girth} of the output spanner (recall that the girth of a graph is the smallest number of edges in a cycle).
One proves that the output spanner always has girth $>2k$, and conversely, any unweighted graph of girth $>2k$ has no $(2k-1)$-spanner except for the graph itself.
Phrased another way, let us write $\gamma(n, 2k)$ for the maximum number of edges in any $n$-node graph of girth $>2k$.
Then the greedy algorithm outputs spanners on $\le \gamma(n, 2k)$ edges, and this bound is best possible.

The only known argument to upper bound the value of $\gamma(n, 2k)$ is the \emph{Moore bounds}, which prove that $\gamma(n, 2k) = O(n^{1+1/k})$ (the girth conjecture~\cite{erdHos1964extremal} is that this is tight, i.e., $\gamma(n,2k) = \Omega(n^{1+1/k})$).
These use a counting argument over the simple $k$-paths of the input graph, where a ``simple'' path is one that does not repeat nodes.
The Moore bounds are proved in two steps.
First, one proves a \emph{counting lemma}: any $n$-node graph of average degree $d \gg 1$ and girth $>2k$ has $n \cdot \Omega(d)^k$ total simple $k$-paths.
Then, one proves a \emph{dispersion lemma}: in a graph of girth $>2k$, no two simple $k$-paths may share endpoints.
Together these imply $n \cdot \Omega(d)^k = O(n^2)$, and the Moore bound are obtained by rearranging this equation.

\paragraph{Background on FT-Spanners.}
In the faulty setting, it is natural to consider the ``FT-Greedy Algorithm,'' which straightforwardly extends the non-faulty greedy algorithm.
That is: consider each edge $(u, v)$ in order of nondecreasing weight, and add $(u, v)$ to the spanner if and only if there is a set of $|F| \le f$ faults such that $\dist_{H \setminus F}(u, v) > (2k-1) \cdot \dist_{G \setminus F}(u, v)$.
% \mdnote{Currently don't describe base greedy algorithm, do we?}
% \gbnote{Base greedy is described early in sec 1, but not in this overview.  Maybe we overview FT-Greedy in the "algorithmic efficiency" section instead of here?} 
Correctness is again easy to show, but the challenge is to control the number of edges in the output spanner.
Informally speaking, the output spanner ought to be ``close'' to a high-girth graph -- not in the sense that its girth is high (the girth of the output spanner could be $3$), but in the sense that the output graph ought to be sparse for a similar reason that high-girth graphs must be sparse.
There are two basic approaches to formalizing this intuition:

\begin{enumerate}
    \item The first analysis of the greedy algorithm \cite{BDPW18} argued that greedy FT-spanners are similar to high-girth graphs in the sense that they are amenable to a generalized version of the Moore bound analysis.
    This led to a rather complicated analysis that took on $\exp(k)$ factors, but which gave optimal dependence on $n$ and $f$ for the size of VFT spanners, assuming the girth conjecture.
    Indeed, arguments of this type can only hope to achieve \emph{conditional} optimality %\mdnote{I think on page 1 we define existential optimality as meeting the Moore bounds.  This is a bit inconsistent (which probably doesn't matter)} 
    on the girth conjecture, since the Moore bounds themselves are only optimal if the girth conjecture holds.

    \item The following analyses~\cite{BP19, DR20, BDR21} took an alternate view that greedy FT-spanners are \emph{structurally similar} to high-girth graphs, in the sense that one can change the spanners into high-girth graphs by sacrificing only a small amount of their density.
    This approach has a number of benefits over the generalized Moore argument: it is far simpler, it removes all extra factors of $k$ in the spanner size, and it lets one ``black-box'' the Moore bounds.  That is, it gives bounds that depend directly on $\gamma(n,2k)$ rather than the Moore bounds for this function, thus establishing optimal spanner size in the VFT setting even if the girth conjecture fails.
    So the message of these papers seemed to be that the structural approach, which only uses the Moore bounds as a black-box to plug into $\gamma(n,2k)$, dominates the approach that opens the black box of the Moore bounds in every important way.
\end{enumerate}

Since a gap remained for EFT spanners, it was still technically unresolved whether Moore or structural arguments are to be preferred in this arena.
This point was discussed in~\cite{BP19}, which contained a technical barrier suggesting that it would be difficult to push the EFT upper bounds any further using the latter structural approach.
Specifically, they suggested the following objects for EFT structural analysis:
\begin{definition} [Edge Blocking Sets \cite{BP19}] \label{def:blocking}
Given a graph $G = (V, E)$, a \emph{$t$-edge blocking set} for $G$ is a set $B \subseteq \binom{E}{2}$ such that, for every cycle $C \in G$ with at most $t$ edges, there exists $(e_1,e_2) \in B$ such that $e_1, e_2 \in C$. 
\end{definition}

One proves the previous-best upper bounds on FT-greedy spanners by (1) observing that the output spanners $H$ from the EFT-greedy algorithm have $2k$-edge blocking sets of size $|B| \le f |E(H)|$, and then (2) proving that \emph{any} graph with a $2k$-edge blocking set of size $|B| \le f |E(H)|$ can only have $O(f^{1-1/k} n^{1+1/k})$ edges.
The second lemma is the key to the structural approach, which appeared again in some form in followup work \cite{DR20, BDR21}.
However, it was further proved in \cite{BP19} that this second lemma is tight: there exist graphs that have a blocking set of size $|B| \le f|E(H)|$ and $\Omega(f^{1-1/k} n^{1+1/k})$ edges.
So, if better EFT upper bounds are to be achieved, they at least need to use a different object from blocking sets, and perhaps depart from this style of argument altogether.

\subsection{Our Approach}
We show, perhaps surprisingly, that the original Moore-based approach seems to dominate the structural one in the context of EFT spanners.
Our analysis is best viewed as a generalization of the Moore bounds,
%and \cite{BDPW18},
and moreover, as an auxiliary result we strengthen the evidence from \cite{BP19} that structural arguments cannot possibly achieve the improved EFT upper bounds obtained in this paper.

First, to sidestep the barrier from \cite{BP19}, we need to change the focus of our analysis.
We use \emph{strong blocking sets}, a natural extension that explicitly takes edge weights into account.
While the original definition of blocking sets is just about edges in cycles, we now require one of the edges to be the \emph{heaviest} edge.  More formally:
\begin{definition} [Strong Blocking Sets] \label{def:strong-blocking}
A \emph{strong $t$-blocking set} is some $B \subseteq \binom{E}{2}$ such that, for any cycle $C$ on at most $t$ edges, there is $(e_1, e_2) \in B$ with $e_1, e_2 \in C$ and such that one of $e_1, e_2$ is the heaviest edge in $C$.
The elements $(e_1, e_2) \in B$ are called \emph{blocks}.
\end{definition}

(We assume that ties between edge weights are broken in some canonical way, and so ``the heaviest edge'' in $C$ is unambiguous.)
It is again easy to prove that the FT-Greedy algorithm produces spanners that have strong blocking sets of size $\le f|E(H)|$.
Since the girth of a graph has nothing to do with its edge weights, we can view this definition as a step away from a purely structural approach, and towards an approach that cares about the specific interaction between the input graph and the FT-greedy algorithm (which considers edges in increasing weight order).

The main new technical idea goes into our generalized dispersion lemma.
The goal is to bound the total number of simple $k$-paths in the graph.
The core of the argument is the following lemma: up to some technical details that we will not discuss here (see Section~\ref{sec:strong-blocking}), due to the existence of a small strong blocking set, it is nearly true that for any two nodes $s, t$ in the output spanner, there is a set of $O(kf)$ edges such that every simple $s \leadsto t$ $k$-path uses one of these edges as its \emph{heaviest} edge.
We may then bound the number of simple $s \leadsto t$ $k$-paths by considering each edge $(u, v)$ in this set and recursively bounding the number of $s \leadsto u$ $j$-paths and the number of $v \leadsto t$ $(k-j-1)$-paths (for different values of $j$).

This recursion turns out to give the strongest bounds when $1 \le j \le k-1$; that is, $(u, v)$ is a \emph{middle} edge along a path, not the first or last one.
This motivates a narrowing of focus in our counting argument from all simple $k$-paths to only those that are \emph{middle-heavy} (have their heaviest edge in a middle position).
In fact, we need our paths to stay middle-heavy over the recursion, as they are repeatedly split over their heaviest edge.
So we focus even more narrowly on a class of \emph{alternating paths}, for which each even-numbered edge is heavier than the two odd-numbered edges on either side, which have this favorable property.
This induces some additional technical work in the counting lemma: standard arguments guarantee the existence of many simple $k$-paths, but we need to rework these arguments to show that a good fraction of these simple $k$-paths are indeed alternating.

\paragraph{Lower Bounds on Structural Analysis.}
%\mdnoteinline{Revise a bit}
%\gbnoteinline{Revised further to basically move explanation before theorem statement.}
Although strong blocking sets move away from structural analyses a bit, it is still reasonable to ask whether our argument can be replaced with a structural one somehow.
Next, we explain our new evidence that the answer to this question is probably ``no'': we show that a structural analysis would by necessity involve solving a long-standing open problem in extremal graph theory.
As mentioned earlier, a key advantage of structural analyses is that they give spanner upper bounds directly in terms of the function $\gamma$.
For example, \cite{BP19} shows that all graphs have EFT or VFT $(2k-1)$ spanners on $O(f^2 \cdot \gamma(n/f, 2k))$ edges.
The Moore bounds are applied as a secondary step to bound $\gamma$, giving the upper bound of $O(f^{1-1/k} n^{1+1/k})$ from Theorem \ref{thm:ftub}.
But if the girth conjecture fails, and there is an even better upper bound for $\gamma$ than the Moore bounds, then the upper bounds from these arguments automatically improve as well.
Conversely, if one can prove \emph{unconditionally} that $f$-VFT spanners need $\Omega(f^{1-1/k} n^{1+1/k})$ edges, this would imply a lower bound on $\gamma$ that would be enough to prove the girth conjecture, which is currently unknown and is a famous and longstanding open problem.

Let us now entertain the hypothetical that one could prove our upper bounds for EFT spanners with a structural argument, arguing that greedy EFT-spanners can be reduced to high-girth graphs while sacrificing only a small amount of density.
For example, analogous to \cite{BP19}, let us imagine that we have a bound of the form $O_k(f \cdot \gamma(n/\sqrt{f}, 2k) + nf)$ edges for odd $k$, which would in turn imply our upper bound of $O_k(f^{1/2 - 1/(2k)} n^{1+1/k} + nf)$ for odd $k$ by applying the Moore bounds to $\gamma$.
Then if one could prove \emph{unconditionally} that $f$-EFT spanners need $\Omega(f^{1/2 - 1/(2k)} n^{1+1/k} + nf)$ edges for fixed odd $k$, it would imply a lower bound on $\gamma$ that would be enough to prove the girth conjecture for fixed odd $k$, which would again be a major breakthrough.

In fact, we can prove a new unconditional bound of this type -- not for all fixed odd $k$, but specifically for $k=7$.  A proof of the following theorem appears in Appendix~\ref{app:lbs}.
\begin{theorem}
For $k=7$ and $f = \Omega(n^{1/11})$, the lower bound on EFT spanners in Theorem \ref{thm:eftlb} holds unconditionally.
\end{theorem}

Thus, a structural argument as explained above would be enough to prove the girth conjecture for $k=7$, which is currently unknown.
Since no new setting of the girth conjecture has been settled in over 50 years, this is likely very hard, and it would still constitute a breakthrough in the area.

% To make this point more concrete, the analogous structural upper bound to \cite{BP19} (if possible) would have the form $O(f \cdot \gamma(n/\sqrt{f}, 14))$ when $k=7$.
% This upper bound, together with the above theorem, would imply:
% $$f \cdot \gamma(n / \sqrt{f}, 14) = \Omega\left( f^{1/2 - 1/14} n^{8/7}\right)$$
% and so
% $$\gamma(n / \sqrt{f}, 14) = \Omega\left(\left( \frac{n}{\sqrt{f}} \right)^{8/7} \right).$$
% Reparametrizing $n' := n/\sqrt{f}$, we get
% $$\gamma(n', 14) = \Omega\left( (n')^{8/7} \right),$$
% confirming the girth conjecture for $k=7$.

% \subsection{Outline}
% As discussed,
% %\mdnote{not necessarily in new version}
% to prove Theorems~\ref{thm:main} and~\ref{thm:main-polytime} we follow a two-step process: we first prove in Section~\ref{sec:greedy-blocking} that the greedy algorithm and its polynomial-time variant (introduced in~\cite{DR20} for vertex fault tolerance) create spanners with small \emph{strong} blocking sets.  This allows us in Section~\ref{sec:strong-blocking} to forget about spanners, and just prove that any graph with small strong blocking sets must be sufficiently sparse.
% %\gbnoteinline{Revisit after merged 1.2}

\section{From Greedy to Strong Blocking} \label{sec:greedy-blocking}

We begin our technical work by proving that the greedy FT-spanner algorithm analyzed in \cite{BDPW18,BP19} gives spanners with small strong blocking sets.

\begin{algorithm}
\caption{Greedy $f$-EFT $(2k-1)$-Spanner Algorithm}
\label{ALG:old}
\begin{algorithmic}
\STATE $\mathbf{function}$ FT-GREEDY$(G = (V,E,w),k,f)$\\
\STATE $H\leftarrow (V,\emptyset,w)$
\FORALL{$(u,v)\in E$ in nondecreasing weight order}
    \IF{there exists a set $F$ of at most $f$ edges such that $d_{H\setminus F}(u,v) > (2k-1) w(u,v)$}
        \item add $(u,v)$ to H
    \ENDIF
\ENDFOR
\RETURN H
\end{algorithmic}
\end{algorithm}

% The following lemma was given in \cite{BP19}. 

% \begin{lemma}
% Any graph H returned by the EFT greedy algorithm with parameters $k,f$ has a $2k$-blocking set of size at most $f|E(H)|$.
% \end{lemma}

%\gbnoteinline{Added tiebreaking details here, which had a nonexistent back pointer from sec 3}
This algorithm implicitly breaks ties between edges of equal weight, to consider one before the other in the main for-loop.
This tiebreaking is arbitrary, but in the sequel it will be convenient to unambiguously refer to \emph{the heaviest} edge among a set of edges, which means the heaviest edge under this tiebreaking (equivalently, the edge considered last by the greedy algorithm).
In particular, in the definition of strong blocking sets, it is required that for each cycle $C$, there is a block $(e_1, e_2)$ with $e_1, e_2 \in C$ and such that either $e_1, e_2$ is the heaviest edge in $C$ \emph{under this tiebreaking}.

The following lemma is a straightforward adaptation of Lemma 3 of~\cite{BP19}, where the same statement was proved for blocking sets rather than strong blocking sets.
%\mdnote{Need to make sure that blocking and strong blocking sets already defined}

\begin{restatable}{lemma}{greedy} \label{lem:greedy-strong-blocking}
Any graph $H$ returned by Algorithm~\ref{ALG:old} is an $f$-EFT $(2k-1)$-spanner that has a strong $2k$-blocking set of size at most $f|E(H)|$.
\end{restatable}
\begin{proof}
See Appendix \ref{app:greedy-blocking}.
\end{proof}

One major downside of this algorithm is its running time.  In each iteration of this algorithm, the algorithm needs to decide the following question: given an edge $(u,v)$, is there a fault set $F\subseteq E$ where $|F|\leq f$ and $d_{H\setminus F}(u,v)\leq (2k-1)\cdot d_{G\setminus F}(u,v)$? Doing this in the obvious way (checking all sets of $f$ or fewer edges) takes time exponential in $f$.  And, unfortunately, it turns out that this question is equivalent to the Length Bounded Cut problem (LBC), which is known to be NP-hard \cite{BEHKSS06}.  

In order to develop a polynomial time algorithm for constructing FT-spanners, \cite{DR20} replaced this exponential time subroutine with one that ran in polynomial time. They designed a simple $(2k-1)$-approximation algorithm for the \emph{unweighted} setting which is essentially the standard frequency approximation of Set Cover.  In particular, when used for the decision variant it has the property that if there is a fault set of size at most $f$ that intersects all $u-v$ paths with at most $2k-1$ hops then the algorithm will return YES, while if all such fault sets have size larger than  $(2k-1)f$ then it will return NO (if neither is true then the algorithm can return either answer).  We call this algorithm $\mathcal A(G, u, v, f, k)$

% \begin{algorithm}
% \caption{Algorithm for LBC($t, \alpha$)}
% \label{ALG:LBC}
% \begin{algorithmic}
% \STATE $F \leftarrow \emptyset$
% \FOR{$i = 1$ to $\alpha+1$}
%     \STATE Run BFS to find a path $P$ of length at most $t$ from $u$ to $v$ in $G \setminus F$ if one exists.
%     \IF{no such $P$ exists} \RETURN YES
%     \ELSE \STATE Add all edges of $P$ to $F$
%     \ENDIF 
% \ENDFOR
% \RETURN NO
% \end{algorithmic}
% \end{algorithm}

A key insight is of~\cite{DR20} was that, even though $\mathcal A$ is designed for the unweighted setting, it can still be used in the weighted setting.  In particular, if we use the weights to set the initial ordering of edges (as in Algorithm~\ref{ALG:old}) but from then on pretend that the graph is unweighted, this modified greedy algorithm still returns a valid $f$-EFT $(2k-1)$-spanner.  This gives the following polynomial time algorithm for constructing EFT spanners.   %Using this algorithm, they provided the following modified greedy algorithm for constructing EFT-spanners in the unweighted setting.

\begin{algorithm}
\caption{Modified Greedy EFT Spanner Algorithm}
\label{ALG:polytime}
\begin{algorithmic}
\STATE $\mathbf{function}$ FT-GREEDY$(G = (V,E),k,f)$\\
\STATE $H\leftarrow (V,\emptyset)$
\FORALL{$(u,v)\in E$ in nondecreasing weight order} 
    \IF{$\mathcal A(H, u, v, f, k)$ returns YES}
    \STATE Add $(u,v)$ to $H$
    \ENDIF
\ENDFOR
\RETURN H
\end{algorithmic}
\end{algorithm}

In order to analyze the sparsity of spanners constructed by this algorithm, \cite{DR20} showed that the blocking set analysis of~\cite{BP19} can be applied with an extra loss of $O(k)$ in the size (due to the approximation).  We show the same, but for our new notion of strong blocking sets.
%More formally, they proved the following.  They provided the following lemma.

% \begin{lemma}
% Any graph H returned by the Modified EFT greedy algorithm with parameters $k,f$ has a $2k$-blocking set of size at most $(2k-1)f|E(H)|$.
% \end{lemma}

% \noindent The proof of this lemma can be easily adapted to show the following.

\begin{restatable}{lemma}{polygreedy} \label{lem:greedy-polytime-strong-blocking}
The graph $H$ returned by Algorithm~\ref{ALG:polytime} is an $f$-EFT $(2k-1)$-spanner that has a strong $2k$-blocking set of size at most $O(kf) \cdot |E(H)|$. %\gbnote{Slight preference for writing $O(kf)$, since we don't need to carry the exact constant forward.}\mdnote{I don't care much, but we should be consistent with previous lemma}
\end{restatable}
\begin{proof}
See Appendix \ref{app:greedy-blocking}
\end{proof}

\section{From Strong Blocking Sets to Extremal Density Bounds} \label{sec:strong-blocking}

In this section, we will prove:
\begin{theorem} \label{thm:tech-main}
Let $f$ be a positive integer, let $k$ be a positive integer, and let $H$ be an $n$-node graph that has a strong $2k$-blocking set $B$ of size $|B| \le |E|f$.
Then
$$|E(H)| = \begin{cases} O\left(k^2 f^{1/2 - 1/(2k)} n^{1+1/k} + nfk\right) & k \text{ is odd}\\
O\left(k^2 f^{1/2} n^{1+1/k} + nfk\right) & k \text{ is even}.
\end{cases}$$
\end{theorem}
%\gbnoteinline{Change parameter $f$ to $\beta$ or something throughout.}
This theorem together with Lemma~\ref{lem:greedy-strong-blocking} directly implies Theorem~\ref{thm:main}.
For the polynomial-time version, note that Lemma~\ref{lem:greedy-polytime-strong-blocking} gives us a polytime algorithm with $|B| = O(kf|E(H)|)$, so we can reparameterize Theorem~\ref{thm:tech-main} by setting the $f$ in the theorem equal to $O(kf)$.  This implies Theorem~\ref{thm:main-polytime}.  % .\mdnote{Need to change lemmas to include correctness}

\subsection{Proof Preliminaries} \label{sec:simplifications}

Our proof will be a counting argument over a special type of path in the graph $H$:
\begin{definition} [Alternating Paths]
An \emph{alternating $k$-path} in $H$ is a path $\pi$ with edge sequence $(e_1, e_2, \dots, e_k)$ 
%\gbnote{This notation means node sequence to me, but not sure what is standard.} 
with the following property: every even-numbered edge is heavier than the two odd-numbered edges adjacent to it.
(When $k$ is even, the last edge $e_k$ is only required to be heavier than $e_{k-1}$.)
\end{definition}

More specifically, we will count \emph{simple unblocked} alternating paths.
As usual, a \emph{simple} path is one that does not repeat nodes.
An \emph{unblocked} path is defined as follows:
%\gbnote{Change: new def.  (Better to call this ``weakly blocked'' or something?)}
\begin{definition} [Blocked Paths]
A path $\pi$ is \emph{blocked} by a strong blocking set $B$ if there is $(e_1, e_2) \in B$ with $e_1, e_2 \in \pi$.
Otherwise, $\pi$ is \emph{unblocked}.
(Note: unlike for cycles, $\pi$ is blocked even if neither $e_1$ nor $e_2$ is the \emph{heaviest} edge in $\pi$.)
\end{definition}

In the remainder of this section we follow the outline discussed in Section~\ref{sec:overview}, giving a counting lemma and a dispersion lemma for simple unblocked alternating $k$-paths.  But we first start by making some useful simplifications, which are all without loss of generality:

\begin{enumerate}
    \item We will use the stronger hypothesis that in the strong blocking set $B$, each edge appears in $\le f$ blocks.
This is without loss of generality for the following reason.
Since every block contains two edges, the average edge participates in $\le 2f$ blocks, and hence a simple counting argument (or Markov's inequality) implies that at most half of the edges participate in at least $4f$ blocks.
Suppose that, while there is an edge that participates in $\ge 4f$ blocks, we delete that edge from the graph and the corresponding blocks from the blocking set.
We thus delete at most half of the graph edges in this way, and we arrive at a graph with a blocking set in which every block participates in $< 4f$ blocks.
We can then reparametrize $f \gets 4f$ while changing our claimed upper bound on spanner size by only a constant factor.

\item We assume that $H$ has average degree $d \geq 32fk$.
In the regime where this does not hold, $H$ has $nd/2 = O(nfk)$ total edges and so the claimed bounds hold already due to the trailing additive term.

\item Additionally, letting $d$ be the average degree in $H$, we will assume that the maximum degree is $O(d)$.
This is justified by the following operation: while there is a node $v$ of degree $\ge 4d$, split it into two nodes $v_1, v_2$, equitably partitioning the edges that use $v$ between its two copies.
Notice that this splitting operation does not create cycles or edges (though it may destroy some cycles), 
%\mdnote{might destroy cycles, right?}
and thus $B$ is still a strong blocking set for the modified graph.
A quick counting argument, standard in prior work (for example, \cite{BDPW18}), shows that one introduces only $O(n)$ new nodes in this way, and hence this again changes our claimed upper bounds on spanner size only by a constant factor.

\item We assume that all edges in $H$ have distinct weights, so that we may unambiguously refer to \emph{the heaviest} or \emph{lightest} edge among a set of edges.
This is justified by simply breaking weight ties in some arbitrary but consistent way, as in our discussion of the FT-greedy algorithm.
\end{enumerate}

\subsection{Counting Lemma for Alternating Paths}

First, we prove a counting lemma for alternating paths.
We will start with some weaker versions of the counting lemma we need, and then gradually bootstrap them into a full one.
In the following lemma statements, an \emph{edge-simple alternating path} is an alternating path that does not repeat edges (but which might repeat nodes).

\begin{lemma} [Weak Counting Lemma] \label{lem:weakcounting}
For any positive integer $k$, any $n$-node graph $\Gamma$ with $\ge kn$ edges has an edge-simple alternating $k$-paths.
\end{lemma}
\begin{proof}
The proof is by induction on $k$.
For the base case, when $k=0$, any individual node may be viewed as an edge-simple alternating $0$-path.

Now we prove the inductive step.
We will first assume that $k$ is odd.
Assume without loss of generality that $H$ has no isolated nodes, and preprocess $\Gamma$ by removing the lightest edge incident to each node.
We remove at most $n$ edges in this way, so at least $(k-1)n$ edges remain.
By the inductive hypothesis, there exists an edge-simple alternating $k-1$ path $\pi$ in the surviving graph.
Let $(u, v)$ be the last edge of $\pi$.
Since $k-1$ is even, we can extend $\pi$ into an alternating $k$-path by appending any edge incident to $v$ that is lighter than $(u, v)$.
Letting $(v, x)$ be the edge removed in the preprocessing, we know that $(v, x)$ is lighter than $(u, v)$, since it was selected to be removed instead of $(u, v)$.
Thus we can extend $\pi$ into an alternating $k$-path by appending $(v, x)$, and since $\pi$ was generated in a graph that did not contain $(v, x)$, it remains edge-simple under this extension.

The case where $k$ is even is similar, except that we remove the \emph{heaviest} edge incident to each node in the preprocessing.
\end{proof}

We next bootstrap into a stronger counting lemma over the same kinds of paths:

\begin{lemma} [Intermediate Counting Lemma] \label{lem:counting}
For any positive integer $k$, any $n$-node graph $\Gamma$ with $\ge 2kn$ edges has $\ge kn$ total edge-simple alternating $k$-paths.
\end{lemma}
\begin{proof}
Consider the following process: pick an edge that is contained in at least one edge-simple alternating $k$-path, remove it, and repeat until there are no more edge-simple alternating $k$-paths.
Since we remove the edge we find in each iteration, there are at least as many edge-simple alternating $k$-paths as there are iterations.
And there are at least $kn$ iterations, since as long as $\ge kn$ edges remain in the graph, Lemma~\ref{lem:weakcounting} implies that there will still be an edge-simple alternating $k$-path.
Hence there are at least $kn$ edge-simple alternating $k$-paths in total.
\end{proof}

Now we bootstrap this once again into the counting lemma we will actually use in our main argument.
While the previous two lemmas hold for arbitrary graphs $\Gamma$, the following one holds specifically for the graph $H$ under consideration, which has a small strong blocking set.
%Recall in this lemma statement that a ``simple'' path is one that does not repeat nodes.
%\gbnote{"Simple" reminder quite a few times by now.  Do we still need it?}
\begin{lemma} [Full Counting Lemma]
The graph $H$ has $n \cdot \Omega(d/k)^k$ total unblocked simple alternating $k$-paths, where $d$ is its average degree.
\end{lemma}
\begin{proof}
In this proof, we will actually count the unblocked edge-simple alternating $k$-paths in $H$.
We notice that, if an edge-simple $k$-path repeats a node, then it must have a cycle of length $\le k$ as a subpath, and this cycle must be blocked by $B$.
Thus every path $\pi$ that is edge-simple, but not simple, is already blocked, and hence
the number of unblocked edge-simple alternating $k$-paths is equal to the number of unblocked simple alternating $k$-paths.
In the following, let $\alpha(H)$ denote the number of unblocked (edge-)simple alternating $k$-paths in $H$ (and we use similar notation for other graphs that will arise in the argument).

Let $H'$ be a random subgraph of $H$ obtained by including each edge with probability $8(d/k)^{-1}$.
So we have
$$\mathbb{E}\left[|E(H')|\right] = |E(H)| \cdot 8(d/k)^{-1} = \frac{nd}{2} \cdot 8(d/k)^{-1} = 4nk.$$

Since each edge is sampled independently, a standard Chernoff bound implies that with high probability $|E(H')| \geq (15/4) nk$.%\mdnote{might want to be more formal} \gbnote{I'm happy with this level.}

Let $B' \subseteq B$ be the strong blocking set for $H'$ obtained by keeping the blocks $(e_1, e_2) \in B$ where both $e_1, e_2$ survive in $H'$.
We have
\begin{align*}
\mathbb{E}\left[|B'|\right] = |B| \cdot 64(d/k)^{-2} &\le |E(H)|f \cdot 64(d/k)^{-2}\\
&= \frac{nd}{2} \cdot f \cdot 64(d/k)^{-2}\\
&= 32 nf k^2 /d\\
&\le nk
\end{align*}
where the last inequality uses our assumption that $d \geq 32fk$.  Hence another standard Chernoff bound implies that $|B'| \leq (5/4)nk$ with high probability. %\mdnote{maybe more formal}
Now, build a subgraph $H''$ from $H'$ by considering each block $(e_1, e_2) \in B'$ and deleting $e_1$.
We thus have that with high probability, 
$$
|E(H'') = |E(H')| - |B'| \ge (15/4)nk - (5/4)nk = (5/2)nk.
$$

By Lemma \ref{lem:counting}, a graph on $\ge 2nk$ edges has $\ge kn$ total edge-simple alternating $k$-paths.
By construction \emph{no} $k$-paths in $H''$ are blocked (since we have removed at least one edge from every block originally in $B$).
Thus, we have that $\alpha(H'') \geq kn$ with high probability, and hence $\mathbb{E}[\alpha(H'')] = \Omega(kn)$.
Since $H'' \subseteq H$, this implies that $\mathbb{E}[\alpha(H')] = \Omega(kn)$.

Let us now compute $\mathbb{E}[\alpha(H')]$ in a different way.
Recall that $H'$ was obtained from $H$ by including each edge with probability $O(d/k)^{-1}$.
So each individual edge-simple $k$-path from $H$ survives in $H'$ with probability $O(d/k)^{-k}$.
So we have
$$\mathbb{E}[\alpha(H')] = \alpha(H) \cdot O(d/k)^{-k}.$$
Putting together these two bounds on $\mathbb{E}[\alpha(H')]$, we get
$$\alpha(H) = \Omega(kn) \cdot \Omega(d/k)^k = n \cdot \Omega(d/k)^k,$$
as claimed.
\end{proof}

\subsection{Dispersion Lemma for Alternating Paths}

In this part, we bound the maximum number of alternating paths that can go between two nodes.
First we prove a useful intermediate lemma which applies to all unblocked $k$-paths (not just alternating):

\begin{lemma} \label{lem:chokeedges}
For any nodes $s, t$, there is a set $F_{s,t}$ containing $|F_{s,t}| \le kf+1$ edges such that every simple unblocked $s \leadsto t$ path of length $\le k$ in $H$ uses an edge in $F_{s, t}$ as its heaviest edge.
\end{lemma}
%\gbnoteinline{We could omit ``alternating'' from this lemma statement, since the proof doesn't need it.}
\begin{proof}
We build $F_{s, t}$ iteratively as follows.
Initially $F_{s, t} := \emptyset$.
In each round, pick the heaviest edge contained in any simple unblocked $s \leadsto t$ path of length $\le k$ whose heaviest edge is not already in $F_{s, t}$, and add this edge to $F_{s, t}$.
Halt once no such path exists.

To get the size bound on $F_{s, t}$, suppose for contradiction that after $kf+1$ rounds there remains a simple unblocked $s \leadsto t$ path $\pi$ of length $\le k$ whose heaviest edge is not already in $F_{s,t}$.
For each $f \in F_{s, t}$, by construction there is a simple unblocked $s \leadsto t$ path $\pi'$ of length $\le k$ that uses $f$ as its heaviest edge, and $f$ is heavier than any edge in $\pi$.
Thus $\pi \cup \pi'$ contains a cycle $C$ of length $\le 2k$, which has $f$ as its heaviest edge.  %\gbnote{More detail on this sentence?  Is this clear?}\mdnote{It required a minute of thought from me, so it's fine as is but it's not a bad idea to expand on it a bit more}
This cycle $C$ must be blocked by $B$, so there is $(f, e) \in B$ with $f, e \in C$.
Since the path $\pi'$ is unblocked, and $f \in \pi'$, it must be that $e \in \pi$.

But there are $\le k$ edges in $\pi$, and there are $kf+1$ edges in $F_{s, t}$, so by the pigeonhole principle one edge in $\pi$ appears in $\ge f+1$ blocks.
This contradicts our initial assumption from Section \ref{sec:simplifications} that each edge appears in $\le f$ total blocks in $B$.
\end{proof}

\begin{lemma} [Dispersion Lemma] \label{lem:dispersion}
For any nodes $s, t$ and positive integer $j \le k$, the number of simple unblocked alternating $s \leadsto t$ paths of length $j$ in $H$ is
$$\begin{cases}
O\left(k^2 f\right)^{(j-1)/2} & j \text{ is odd}\\
O\left(k^2 f\right)^{j/2} & j \text{ is even} \end{cases}$$
\end{lemma}
\begin{proof}
The proof is by strong induction over $j$.
The base cases are when $j=0$, in which case there is at most one $0$-path (when $s=t$), and also when $j=1$, in which case only a single $(s, t)$ edge can be a $1$-path.

For the inductive step, we apply Lemma \ref{lem:chokeedges} to find a set of $|F_{s, t}| =O(kf)$ edges for which every simple unblocked alternating $s \leadsto t$ path of length $k$ has its heaviest edge in $F_{s, t}$.
The alternating $s \leadsto t$ paths may then be partitioned into $O(k^2 f)$ equivalence classes, by (1) which edge in $F_{s, t}$ is their heaviest, and (2) in which of the $j \le k$ possible positions in the path the heaviest edge occurs.

We now bound the size of each individual equivalence class using the inductive hypothesis.
Consider the equivalence class of simple unblocked alternating paths whose heaviest edge is $(u, v) \in F_{s, t}$, and which use $(u, v)$ in the $i^{th}$ position.
Notice that each such path may be viewed as the concatenation of a simple unblocked $s \leadsto u$ alternating path of length $i-1$, and then the edge $(u, v)$, and then a simple unblocked $v \leadsto t$ alternating path of length $j - i$.
So the number of such paths may be upper bounded by the product of the number of simple unblocked $s \leadsto u$ alternating paths of length $i-1$ and the number of simple unblocked alternating $v \leadsto t$ paths of length $j-i$.

We have that $i$ is always even, since the heaviest edge in any alternating path is even-numbered.
So $i-1$ is odd, and using the strong inductive hypothesis, the number of $s \leadsto u$ simple unblocked alternating paths of length $i-1$ is $O(k^2 f)^{(i-2)/2}$.
The parity $j-i$ depends on the parity of $j$, so we need to split into two cases:
\begin{itemize}
    \item if $j$ is even, and so $j-i$ is even, by the strong inductive hypothesis the number of simple unblocked alternating $v \leadsto t$ paths is $O(k^2 f)^{(j-i)/2}$.
    In this case the product is bounded by
    $$O(k^2 f)^{(i-2)/2} \cdot O(k^2 f)^{(j-i)/2} = O\left(k^2 f\right)^{(j-2)/2}.$$
    Since there are $O(k^2 f)$ total equivalence classes, the total number of simple unblocked alternating $s \leadsto t$ $j$-paths is
    $$O(k^2 f) \cdot O\left(k^2 f\right)^{(j-2)/2} = O\left(k^2 f\right)^{j/2},$$
    as claimed.
    
    \item if $j$ is odd, and so $j-i$ is odd, by the strong inductive hypothesis the number of simple unblocked alternating $v \leadsto t$ paths is $O(k^2 f)^{(j-i-1)/2}$.
    In this case the product is bounded by
    $$O(k^2 f)^{(i-2)/2} \cdot O(k^2 f)^{(j-i-1)/2} = O\left(k^2 f\right)^{(j-3)/2}.$$
    Since there are $O(k^2 f)$ total equivalence classes, the total number of simple unblocked alternating $s \leadsto t$ $j$-paths is
    $$O(k^2 f) \cdot O\left(k^2 f\right)^{(j-3)/2} = O\left(k^2 f\right)^{(j-1)/2},$$
    also as claimed. \qedhere
\end{itemize}
\end{proof}

We have stated this dispersion lemma for general $j$, for the sake of its inductive proof, but we will only need to use the setting $j=k$ in the following.

\subsection{Putting it Together}

Now we prove Theorem \ref{thm:tech-main}.
Let $d$ be the average degree of $H$.
By our counting lemma, $H$ has at least $n \cdot \Omega(d/k)^k$ total simple unblocked alternating $k$-paths.
To apply our dispersion lemma, let us first consider the case where $k$ is odd.
Then there are $O\left(k^2 f\right)^{(k-1)/2}$ total simple unblocked alternating $k$-paths between any given pair of nodes in $H$, for an upper bound of $n^2 \cdot O\left(k^2 f\right)^{(k-1)/2}$
on the total number of simple unblocked alternating $k$-paths in the entire graph.
We can put this upper and lower bound together and rearrange as follows:
\begin{align*}
    & n \cdot \Omega(d/k)^k = n^2 \cdot O\left(k^2 f\right)^{(k-1)/2}\\
    \implies & d/k = n^{1/k} O\left(k^2 f\right)^{1/2 - 1/(2k)}\\
    \implies &|E(H)| = \frac{nd}{2} = O\left( k^2 n^{1+1/k} f^{1/2 - 1/(2k)} \right).
\end{align*}

In the case where $k$ is even, we proceed similarly: the total number of simple unblocked alternating $k$-paths in the entire graph is at most
$n^2 \cdot O\left(k^2 f\right)^{k/2},$
and so we put this upper bound together with the same lower bound as before to get:
\begin{align*}
    & n \cdot \Omega(d/k)^k = n^2 \cdot O\left(k^2 f\right)^{k/2}\\
    \implies & d/k = n^{1/k} O\left(k^2 f\right)^{1/2}\\
    \implies & |E(H)| = \frac{nd}{2} = O\left( k^2 n^{1+1/k} f^{1/2} \right).
\end{align*}

The additional additive $+nfk$ term in Theorem \ref{thm:tech-main} is needed to justify the assumption that $d = \Omega(fk)$, which was used to prove our counting lemma.

\bibliography{refs}

\appendix

\section{Figures \label{app:figures}}

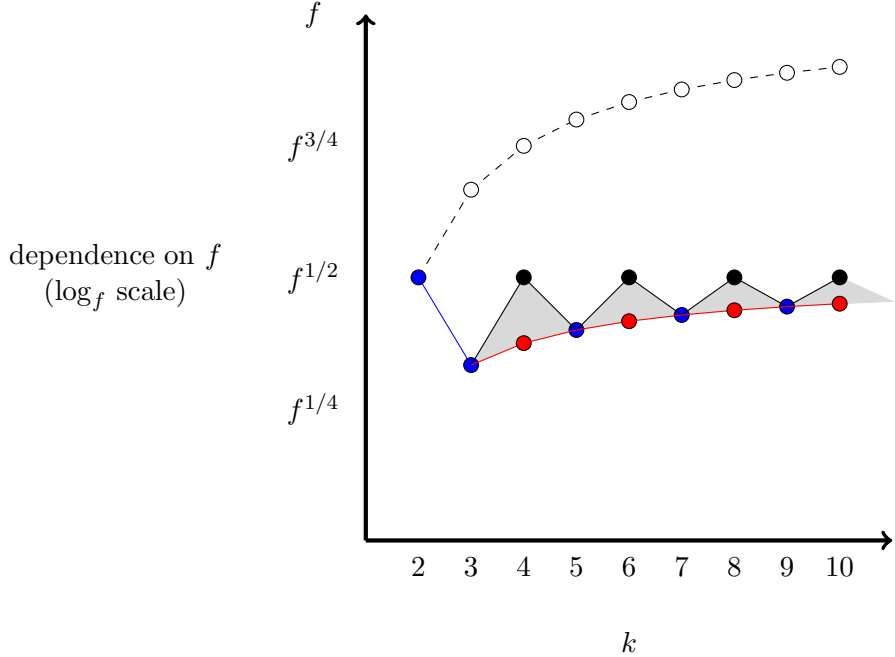
\begin{figure} [h!]
\begin{center}
\begin{tikzpicture} [scale=0.7]
%axes
\draw [ultra thick, ->] (0, 0) -- (10, 0) node [midway, below=30pt] {$k$};
\draw [ultra thick, ->, align=center] (0, 0) -- (0, 10) node [midway, left=50pt] {dependence on $f$\\ ($\log_f$ scale)};

%y axis labels
\node at (-1, 2.5) {$f^{1/4}$};
\node at (-1, 5) {$f^{1/2}$};
\node at (-1, 7.5) {$f^{3/4}$};
\node at (-1, 10) {$f$};

%old
\foreach \k in {2,...,9}{
    \draw [dashed] ({\k - 1, {10-10/(\k)}}) -- ({\k, {10-10/(\k+1)}});
}
\foreach \k in {2,...,9}{
    \draw [fill=white] ({\k, {10-10/(\k+1)}}) circle [radius=4pt];
}

%k=2 exceptional
\draw [fill=blue] (1, 5) circle [radius=4pt];
\draw [blue] (1,5)--(2,{10/3});
\node at (1, -0.5) {$2$};

%gap shading
\foreach \k in {1,...,4}{
    \draw [gray!30, fill=gray!30] ({2*\k}, {5 - 5/(2*\k+1)}) -- ({2*\k + 1}, 5) -- ({2*\k + 2}, {5 - 5/(2*\k+3)}) -- ({2*\k + 1}, {5 - 5/(2*\k+2)}) -- cycle;
}

%odd k
\foreach \k in {1,...,4}{
    \draw [fill=blue] ({2*\k}, {5 - 5/(2*\k + 1)}) circle [radius=4pt];
    \pgfmathtruncatemacro{\axislabel}{2*\k +1};
    \node at ({2*\k}, -0.5) {\axislabel};
    %connecting lines
    \draw ({2*\k}, {5 - 5/(2*\k + 1)}) -- ({2*\k + 1}, 5);
    \draw [red] ({2*\k}, {5 - 5/(2*\k + 1)}) -- ({2*\k + 1}, {5 - 5/(2*\k+2)});
}

%other even k
\foreach \k in {2,...,5}{
    \draw [fill=black] ({2*\k - 1}, 5) circle [radius=4pt];
    \draw [fill=red] ({2*\k - 1}, {5 - 5/(2*\k)}) circle [radius=4pt];
    \pgfmathtruncatemacro{\axislabel}{2*\k};
    \node at ({2*\k - 1}, -0.5) {\axislabel};
}

%more connecting lines
\foreach \k in {2,...,4}{
    \draw ({2*\k - 1}, 5) -- ({2*\k}, {5 - 5/(2*\k + 1)});
}
\foreach \k in {2,...,4}{
    \draw [red] ({2*\k - 1}, {5 - 5/(2*\k)}) -- ({2*\k}, {5 - 5/(2*\k + 1)});
}

\end{tikzpicture}
\end{center}
\caption{Dependence of our new EFT upper bounds on $f$: $f^{1/2}$ for even $k$ and $f^{1/2 - 1/(2k)}$ for odd $k$.  Blue points are where our $f$ dependence is tight, and black/red points represent non-matching upper/lower bounds.  The shaded regions represent knowledge gaps, and the white points are the previous best $f$ dependence ($f^{1-1/k}$) that was known before this paper, in all of \cite{BDPW18, BP19, DR20, BDR21}.}
\label{fig:results}
\end{figure}

\section{Unconditional Lower Bounds for EFT Spanners \label{app:lbs}}

In this section, we point out an extension to the lower bound of \cite{BDPW18} for EFT spanners:
\begin{theorem} \label{thm:uncondeftlb}
For $k=7$ and $f=\Omega(n^{1/11})$, there is an infinite family of $n$-node graphs for which any $f$-EFT $(2k-1 = 13)$-spanner has $\Omega\left( f^{3/7} n^{8/7} \right)$ edges.
\end{theorem}

This is the same bound as in Theorem \ref{thm:eftlb}, but without the need to condition on the girth conjecture despite the fact that the girth conjecture is still unproved for the $k=7$ case.
Our construction is based on an object in incidence geometry called a Ree-Tits Octagon (see \cite{Wilson10} for discussion and a relatively simple construction).
This is a kind of incidence structure, meaning it is a collection of subsets (called ``lines'') of the universe $[n]$ (called ``points'').
For integers $k \ge 2$, let us say that a $k$-gon is a circularly-ordered collection of $k$ pairwise-distinct points $(v_0, v_1, \dots, v_k=v_0)$ such that, for any two adjacent points $v_i, v_{i+1}$, there exists a line that contains both points.
The properties are:

\begin{theorem} [Ree-Tits Octagon, see e.g.\ \cite{Wilson10}]
For infinitely many values of $n$, there is an incidence structure consisting of $\Theta(n^{10/11})$ lines over the points $[n]$ such that:
\begin{itemize}
\item The structure has no $k$-gon for any $2 \le k \le 7$,
\item Every point is contained in $\Theta(n^{1/11})$ lines, and
\item Every line contains $\Theta(n^{2/11})$ points.
\end{itemize}
\end{theorem}

The structure has some other interesting properties as well, but these are the properties relevant to our proof.
We can associate this structure to an incidence graph, which has the following properties:
\begin{corollary} [Ree-Tits Octagon Incidence Graph] \label{cor:rtoincidence}
For infinitely many values of $n$, there is a bipartite graph with $n$ nodes on one side, $\Theta(n^{10/11})$ nodes on the other side, $\Theta(n^{12/11})$ total edges, and girth $>15$.
\end{corollary}
\begin{proof}
Build the incidence graph associated to the Ree-Tits Octagon.
This is a bipartite graph with $n$ nodes on the left side, corresponding to the points $[n]$, and $\Theta(n^{10/11})$ nodes on the right side corresponding to the lines.
We put an edge from a point $x$ on the left to a line $\ell$ on the right iff $x \in \ell$.
The number of edges is immediate from the properties of the Ree-Tits Octagon, and we notice that a $2k$-cycle in the incidence graph corresponds to a $k$-gon in the original incidence structure.
Thus the graph has no cycles of length $4, 6, 8, 10, 12$, or $14$, and since it is bipartite it has no odd-length cycles either.
So its girth is $>15$.
\end{proof}

Now we can prove Theorem \ref{thm:uncondeftlb}; this is the method of \cite{BDPW18}, but applied to a different starting graph.
Starting with a graph $G = (V, E)$ from Corollary \ref{cor:rtoincidence}, copy each node $v$ on the smaller side of the bipartite graph $n^{1/11} \cdot (f/n^{1/11})^{1/2}$ times, and replace each edge $(u, v)$ with an edge $(u, v_i)$ for every copy $v_i$.
We similarly copy each node on the larger side of the bipartite graph $(f/n^{1/11})^{1/2}$ times in the same way.
(We assume for convenience that these quantities are integers, which affects our bounds only by lower-order terms.)
The number of nodes in the modified graph is
$$\Theta\left(n \cdot \left(\frac{f}{n^{1/11}}\right)^{1/2} \right) = \Theta\left(n^{21/22} \cdot f^{1/2} \right)$$
on the left side, and
$$\Theta\left(n^{10/11} \cdot n^{1/11} \cdot \left(\frac{f}{n^{1/11}}\right)^{1/2} \right) = \Theta\left(n^{21/22} \cdot f^{1/2} \right)$$
on the right side as well.
The number of edges in the modified graph is
$$|E(G)| = \Theta\left(n^{12/11} \cdot n^{1/11} \cdot \frac{f}{n^{1/11}} \right) = \Theta\left(f \cdot n^{12/11}\right).$$
Writing $N := n^{21/22} \cdot f^{1/2}$, so that there are $\Theta(N)$ total nodes, we can then reparametrize the number of edges as
$$|E(G)| = \Theta\left( f^{3/7} \cdot N^{8/7} \right).$$
Notice that this is exactly our claimed lower bound.
Next, we argue that $G$ is the only $f$-EFT spanner of itself, and thus no edges can be removed.
To see this, let us consider an edge $(u_i, v_j)$, which is a copy of an edge $(u, v)$ from the original graph.
The total number of copies of $(u, v)$ is the product of the number of copies of $u$ by the number of copies of $v$, which is
$$n^{1/11} \cdot \left( \frac{f}{n^{1/11}} \right)^{1/2} \cdot \left( \frac{f}{n^{1/11}} \right)^{1/2} = f.$$
So we may let $F$ be all other copies of $(u, v)$, besides the edge $(u_i, v_j)$ under consideration, and we have $|F| \le f$.
With these definitions of $G, F$, we have:
\begin{lemma}
In $G \setminus F$, there is no cycle $C$ of length $\le 15$ that includes the edge $(u_i, v_j)$.
\end{lemma}
\begin{proof}
Suppose for contradiction that such a cycle $C$ exists.
Let $\phi$ be the homomorphism from $G \setminus F$ to the original Ree-Tits octagon incidence graph $R$, which maps each node $x_i$ back to the original node $x$ from which it was copied.
Let $\phi(C)$ be the nodewise image of $C$, which is a closed walk in $R$.
Notice that $\phi(C)$ only contains the edge $(u, v)$ once, since we removed all copies of $(u, v)$ in $G$ except for $(u_i, v_j)$.

We then modify $\phi(C)$ as follows: while there exists a node $x$ that occurs at least twice in $C$, replace $\phi(C)$ by its $x \leadsto x$ subwalk that contains the single copy of $(u, v)$.
Once this process halts, we have a simple circularly-ordered walk that contains the edge $(u, v)$ only once, which must therefore be a cycle in $R$ of length $\ge 3$ and $\le 15$.
This contradicts that $R$ has girth $>15$, completing the proof.
\end{proof}

It follows from this lemma that the shortest $u_i \leadsto v_j$ path in $G \setminus (F \cup \{(u_i, v_j)\})$ has length $>14$.
Thus, we must keep $(u_i, v_j)$ in any $f$-EFT $14$-spanner.
Since $(u_i, v_j)$ was an arbitrary edge, this means $G$ is the only $f$-EFT $14$-spanner of itself, completing the proof.

\section{Missing Proofs from Section \ref{sec:greedy-blocking}} \label{app:greedy-blocking}

% \begin{lemma} %\label{lem:greedy-strong-blocking}
% Any graph $H$ returned by Algorithm~\ref{ALG:old} is an $f$-EFT $(2k-1)$-spanner that has a strong $2k$-blocking set of size at most $f|E(H)|$.
% \end{lemma}
\greedy*
\begin{proof}
The fact that $H$ is an $f$-EFT $(2k-1)$-spanner is straightforward from the for-loop; see~\cite{BDPW18} for a formal proof.

We now prove that the output spanner $H$ has a small strong $2k$-blocking set.
For each edge $e=(u,v)\in E(H)$, let $F_e$ denote the set of edges that forced the algorithm to add $e$; i.e., the set with $|F_e| \leq f$ such that $d_{H\setminus F_e}(u,v)>(2k-1)\cdot w(e)$ when $e$ is added to $H$. Let
$$B:=\left\{(x,e) \mid e\in E(H),x\in F_e \right\}.$$
Since $|F_e|\leq f$ for all $e$, we have that $|B|\leq f|E(H)|$.
We can now show that $B$ is a strong $2k$-blocking set for $H$.
Let $C$ be any cycle on $\leq 2k$ edges in the final graph $H$ and let $e=(u,v)$ be the last edge in $C$ considered by the greedy algorithm.
By construction there is a $u \leadsto v$ path (through $C$) of total weight $\leq (2k-1)\cdot w(e)$ when $e$ is added to $H$, and so some edge $x\in C\setminus \{(u,v)\}$ must be included in $F_e$. Thus $(x,e)\in B$.  Moreover, since $e$ is the last edge in $C$ added to $H$ by the algorithm it is the heaviest edge in $C$ (under tiebreaking), as the algorithm considers edges in non-decreasing weight order.  Hence $B$ is a strong $2k$-blocking set.
\end{proof}

% \begin{lemma} %\label{lem:greedy-polytime-strong-blocking}
% The graph $H$ returned by Algorithm~\ref{ALG:polytime} is an $f$-EFT $(2k-1)$-spanner that has a strong $2k$-blocking set of size at most $O(kf) \cdot |E(H)|$. %\gbnote{Slight preference for writing $O(kf)$, since we don't need to carry the exact constant forward.}\mdnote{I don't care much, but we should be consistent with previous lemma}
% \end{lemma}
\polygreedy*
\begin{proof}
The fact that $H$ is an $f$-EFT $(2k-1)$-spanner is again straightforward, and was proved formally in~\cite{DR20}.

We now prove that the output spanner $H$ has a small strong $2k$-blocking set.
For each edge $e=(u,v)\in E(H)$, let $H_e\subseteq H$ denote the subgraph of $H$ containing only the edges considered before $e$ by the algorithm.
Since $e$ was added to $H$ by Algorithm~\ref{ALG:polytime}, $\mathcal A$ must have returned YES. Thus there is some set $F_e\subseteq E\setminus (u,v)$ with $|F_e|\leq f(2k-1)$ such that $\hat d_{H_e\setminus F_e}(u,v)>2k-1$ (where we use $\hat d$ to denote the \emph{unweighted} distance).
Let
$$B:=\left\{(x,e) \mid e\in E(H),x\in F_e\right\}.$$
Since $|F_e|\leq f(2k-1)$ we have that $|B|\leq f(2k-1)|E(H)|$. We can now show that $B$ is a strong $2k$-blocking set for $H$.
Let $C$ be any cycle in $H$ on $\leq 2k$ edges, and let $e=(u,v)$ be the last edge in this cycle added to $H$.
By construction there is a $u \leadsto v$ path (through $C$) with at most $2k-1$ hops when $e$ is added to $H$, and so some edge $x\in C\setminus \{(u,v)\}$ must be included in $F_e$. Thus $(x,e)\in B$.
Moreover, since the algorithm considers the edges in nondecreasing weight order, we know that $e$ is the heaviest edge in $C$ (under tiebreaking).
Thus $B$ is a strong $2k$-blocking set for $H$.
\end{proof}

\end{document}